\documentclass{IEEEtran}
\usepackage{cite}
\usepackage{amsmath,amssymb,amsfonts, amsthm}
\usepackage{algorithmic}
\usepackage{graphicx}
\usepackage{textcomp}
\usepackage{algorithm}
\usepackage{hyperref}
\hypersetup{hidelinks}
\usepackage{mathtools}
\usepackage{booktabs}
\usepackage{siunitx}
\usepackage{xcolor}
\usepackage{csquotes} 
\usepackage{booktabs}

\MakeOuterQuote{"}

\def\BibTeX{{\rm B\kern-.05em{\sc i\kern-.025em b}\kern-.08em
    T\kern-.1667em\lower.7ex\hbox{E}\kern-.125emX}}

\newtheorem{theorem}{Theorem}

\newtheorem{corollary}{Corollary}
\newtheorem{lemma}{Lemma}
\newtheorem{proposition}{Proposition}
\newtheorem{definition}{Definition}

\newcommand{\ie}{i.e.}
\newcommand{\eg}{e.g.}

\newcommand{\reals}{\mathbb{R}}
\newcommand{\SO}[1]{SO(#1)}

\newcommand{\SE}[2]{SE_{#2}(#1)}

\newcommand{\liealg}[1]{\mathfrak{#1}}
\newcommand{\Lmap}[1]{\mathcal{L}_{\liealg{g}}(#1)}

\newcommand{\bfzero}{\mathbf{0}}

\newcommand{\obset}[1]{\mathcal{S}_{\{#1\}}}

\DeclareMathOperator{\im}{Im}

\DeclareMathOperator*{\argmin}{argmin}

\newcommand{\bfa}{\mathbf{a}}

\newcommand{\bfA}{\mathbf{A}}
\newcommand{\bfb}{\mathbf{b}}

\newcommand{\bfB}{\mathbf{B}}

\newcommand{\bfd}{\mathbf{d}}
\newcommand{\bbfd}{\bar{\bfd}}

\newcommand{\bfF}{\mathbf{F}}
\newcommand{\bff}{\mathbf{f}}

\newcommand{\bfG}{\mathbf{G}}

\newcommand{\bfg}{\mathbf{g}}
\newcommand{\bfh}{\mathbf{h}}
\newcommand{\bfH}{\mathbf{H}}

\newcommand{\bbfH}{\bar{\mathbf{H}}}

\newcommand{\bfI}{\mathbf{I}}

\newcommand{\calJ}{\mathbf{\mathcal{J}}}
\newcommand{\bfK}{\mathbf{K}}

\newcommand{\bfL}{\mathbf{L}}

\newcommand{\bfN}{\mathbf{N}}
\newcommand{\hbfN}{\hat{\bfN}}

\newcommand{\bfn}{\mathbf{n}}

\newcommand{\bfp}{\mathbf{p}}

\newcommand{\bfP}{\mathbf{P}}

\newcommand{\bfQ}{\mathbf{Q}}

\newcommand{\bfR}{\mathbf{R}}

\newcommand{\bfS}{\mathbf{S}}

\newcommand{\bfu}{\mathbf{u}}

\newcommand{\bfv}{\mathbf{v}}

\newcommand{\bfw}{\mathbf{w}}

\newcommand{\bfx}{\mathbf{x}}

\newcommand{\bbfx}{\bar{\bfx}}

\newcommand{\bfy}{\mathbf{y}}
\newcommand{\bbfy}{\bar{\bfy}}

\newcommand{\bfz}{\mathbf{z}}

\newcommand{\bochi}{\protect\raisebox{1pt}{$\boldsymbol{\chi}$}}
\newcommand{\hbochi}{\hat{\bochi}}

\newcommand{\boeta}{\boldsymbol{\eta}}

\newcommand{\boomega}{\boldsymbol{\omega}}

\newcommand{\boxi}{\boldsymbol{\xi}}
\newcommand{\bodelta}{\boldsymbol{\delta}}
\newcommand{\bozeta}{\boldsymbol{\zeta}}

\newcommand{\EKF}{\textcolor{black}{EKF}}
\newcommand{\IEKF}{\textcolor{black}{IEKF}}
\newcommand{\IIEKF}{\textcolor{black}{IterIEKF}}
\newcommand{\IterEKF}{\textcolor{black}{IterEKF}}
\newcommand{\LGIEKF}{\textcolor{black}{LG-IterEKF}}

\newcommand{\LGEKF}{\textcolor{black}{LG-EKF}}

\begin{document}

\title{Iterated Invariant Extended Kalman Filter (\IIEKF{})}
\author{
    Sven Goffin, 
    Axel Barrau, 
    Silvère Bonnabel, 
    Olivier Brüls, 
    and 
    Pierre Sacré
    \thanks{S.\ Goffin is a FRIA grantee of the Fonds de la Recherche Scientifique - FNRS.}
    \thanks{S.\ Goffin and P.\ Sacré are with the Department of Electrical Engineering and Computer Science, University of Liège, Belgium (sven.goffin@uliege.be; p.sacre@uliege.be).}
    \thanks{A.\ Barrau is with OFFROAD, France (axel@offroad.works).}
    \thanks{S.\ Bonnabel is with MINES Paris, PSL Research University, Centre for Robotics, France (silvere.bonnabel@mines-paristech.fr).}
    \thanks{O.\ Brüls is with the Department of Aerospace and Mechanical Engineering, University of Liège, Belgium (o.bruls@uliege.be).} 
}

\maketitle


\begin{abstract}
We study the mathematical properties of the Invariant Extended Kalman Filter (\IEKF{}) when iterating on the measurement update step, following the principles of the well-known Iterated Extended Kalman Filter. This iterative variant of the \IEKF{} (\IIEKF{}) systematically improves its accuracy through Gauss-Newton-based relinearization, and exhibits additional theoretical properties, particularly in the low-noise regime, that resemble those of the linear Kalman filter.  We apply the proposed approach to the problem of estimating the extended pose of a crane payload using an  inertial measurement unit. Our results suggest that the \IIEKF{} significantly outperforms the \IEKF{} when measurements are highly accurate.
\end{abstract}


\section{Introduction}
\label{sec:intro}

In the field of state estimation and observer design, the Extended Kalman Filter (\EKF)  is one of the most widespread methods used in practice. 
However,   it is based on linearization, and the associated errors  were early recognized as possibly degrading  performance~\cite{schmidt1981kalman}. This motivated the development of iterative filtering algorithms such as the iterated \EKF{} (\IterEKF{})~\cite{bell1993iterated, zhao2016robust, bloesch2017iterated}, iterated \EKF{} on Lie groups (\LGIEKF{})~\cite{bourmaud2016intrinsic, liu2023general,chahbazian2021laplace}, and smoothers~\cite{van2020invariant}, where the discrepancy between the nonlinear output function and its first-order approximation is reduced, by refining the operating point.  

On another note, observer design has benefited from geometric approaches over the past two decades. They may be traced back to attitude and pose estimation~\cite{bonnabel2005invariant, mahony2008nonlinear, cohen2020navigation}. When turning to the more general problem of inertial navigation and localization, the Invariant Extended Kalman Filter (\IEKF)~\cite{barrau2016invariant,barrau2018invariant} has become a key alternative to the standard \EKF{}. Its theoretical properties include  convergence guarantees~\cite{barrau2016invariant},  consistency properties in the presence of unobservability~\cite{barrau2018invariant}, see also related works~\cite{mahony2017geometric,  van2019geometric, mahony2021homogeneous}, and have led to applications in various fields, \eg, ~\cite{barrau2018invariant, van2020invariant, hartley2020contact, pavlasek2021invariant, wu2017invariant,   heo2018consistent, mahony2017geometric} and in the industry~\cite{barrau2018invariant}.  The field of observers has benefited from the introduction of   new groups brought by the \IEKF{} theory,  see \eg,~\cite{wang2020hybrid, hashim2021gps,van2021autonomous}, namely the groups $SE_2(3)$ and $SE_k(d)$ introduced in~\cite{barrau2015non,barrau2016invariant} and their relation to the navigation equations, see~\cite{barrau2022geometry} for a recent perspective. 
The equivariant observer framework, see~\cite{mahony2021equivariant}, is closely related.

In this technical note, we introduce the Iterated Invariant Extended Kalman Filter (\IIEKF{}), a refinement of the~\IEKF~\cite{barrau2018invariant} that leverages the Gauss-Newton (GN) method  to enhance its measurement update step, and which systematically improves its accuracy.   Our main contribution is then a comprehensive analysis of the \IIEKF{} properties, revealing rare properties in nonlinear estimation that remind of the linear case,  and offering new tools and  insights into the theory of invariant filtering.   Finally, to highlight the practical relevance of our approach, we apply the \IIEKF{} to a problem of engineering interest, and show  it outperforms the \IEKF{}  when observations are very accurate (low measurement noise).

 While the ultimate goal of invariant filtering is to recover properties of the linear Kalman filter in a nonlinear setting, and to identify    systems of interest that lend themselves to this goal, a parallel line of work has been concerned with  how to accommodate in general the geometric nature of the state space, without seeking further properties. In particular, the "Lie-Group EKF"  (\LGEKF{}) from \cite{bourmaud2013discrete} proposes a generic intrinsic version of the EKF on  Lie groups, that accounts for geometric features such as curvature. An iterated version, called  \LGIEKF{} was proposed in \cite{bourmaud2016intrinsic}.  See, \eg, 
 \cite{he2021kalman} and  applications to inertial-lidar navigation \cite{xu2022fast} for recent references in this field.  When tailored to the specific class of problems we consider, namely by using the groups from invariant filtering and the choice of errors that it advocates, the \LGIEKF{} (that we call then \emph{adapted})  is close to the proposed \IIEKF{}, although slightly different. As a result, we can prove  the properties  of the \IIEKF{}    carry over to this adapted \LGIEKF{}. This novel result constitutes a secondary contribution.  


Section~\ref{sec:iterated_IEKF}  introduces the proposed \IIEKF{} algorithm for both noisy and noise-free measurements. 
Section~\ref{sec:properties} derives   theoretical properties of the \IIEKF{} in the context of noise-free measurements. We prove they carry over to the \LGIEKF{} from \cite{bourmaud2016intrinsic} under some conditions.  
Section~\ref{sec:crane} evaluates in simulations the performance of the proposed algorithm compared to other state-of-the-art filters for estimating the extended pose (orientation, velocity, and position) of a crane hook, equipped with an Inertial Measurement Unit (IMU). 

In the following, the acronym \IEKF{} is reserved for the invariant \EKF{}, whereas \IterEKF{} is used for the iterated \EKF{}.

\section{The Iterated Invariant \EKF{} (\IIEKF{})}\label{sec:iterated_IEKF}
In this section, we first recall the equations of the \IEKF{} and then introduce the \IIEKF{}, its iterated version.

\subsection{The IEKF equations}
The invariant framework assumes the state $\bochi_k$ is an element of a matrix Lie group $G \subset GL_N(\reals)$ having dimension~$n$, where $GL_N(\reals)$ denotes the group of $N\times N$ invertible matrices. An example of such a state $\bochi$ is given below, in~\eqref{chi:def:eq}. The exponential map of~$G$ is defined as $\exp_G(\cdot) := \exp_m(\Lmap{\cdot})$, where $\exp_m(\cdot)$ is the matrix exponential and $\Lmap{\cdot}$ is the bijective linear map identifying the Lie algebra $\liealg{g}$ with $\reals^n$, see \eg, \cite{barrau2016invariant}. The \IEKF{} comes in two versions, the left- or the right-invariant version, depending on the form of the observations, see~\cite{barrau2018invariant}. This note focuses on the left-\IEKF{}. Transposition to the right-\IEKF{} is straightforward.  

Consider the following nonlinear system in discrete time:
\begin{subequations} \label{inv:system}
    \begin{align}
        \bochi_{k+1} &= \bff(\bochi_k, \bfu_k, \bfw_k),\label{inv:dyna}\\
        \bfy_{k} &= \bochi_k\bfd_k + \bfn_{k},\label{inv:measurements}
    \end{align}
\end{subequations}
with $\bfu_k \in \reals^m$  a control input,   
$\bff:G\times \reals^m \times \reals^b \to G$ the function describing the system dynamics, 
$\bfw_k \sim \mathcal{N}(\bfzero, \bfQ_k)$ an unknown process noise with $\bfQ_k\in \reals^{b\times b}$, $\bfy_{k} \in \reals^N $   the observation that consists of partial and noisy measurements of the state,
$\bfd_k \in \mathbb{R}^N$   a known vector, and 
$\bfn_k \sim \mathcal{N}(\bfzero, \bfN_k)$  an unknown measurement noise with $\bfN_k \in \reals^{N \times N}$.

The (left) \IEKF{} assumes that the state follows a concentrated Gaussian distribution on $G$~\cite{chirikjian2009stochastic, wang2006error, wolfe2011bayesian, bourmaud2015continuous}. At time $k$, letting~$l$ denote current time index $k$, or previous index $k-1$,  
\begin{equation} \label{eq:concentrated_gaussian} 
    \bochi_k = \hbochi_{k\mid l} \exp_G(\boxi_{k\mid l}), \quad \text{where } \boxi_{k\mid l} \sim \mathcal{N}(\bfzero, \bfP_{k\mid l}), 
\end{equation} 
where $\hbochi_{k\mid l} \in G$ is the (best) estimate of the state, and the linearized error $\boxi_{k\mid l} \in \mathbb{R}^n$ is a centered Gaussian   with   covariance matrix $\bfP_{k\mid l} \in \reals^{n \times n}$. 
Using this model, the \IEKF{} linearizes the equations of the system~\eqref{inv:system} at the current estimate, leading to the following update and propagation:
\begin{equation}\label{eq:IEKF_equations}
    \begin{tabular}{ll}
        Upd. & $\left\{\begin{aligned}
            \bfz_k &= \hbochi_{k\mid k-1}^{-1} \bfy_k - \bfd_k,\\
            \hbfN_k &= \hbochi_{k\mid k-1}^{-1} \bfN_k (\hbochi_{k\mid k-1}^{-1})^T,\\
            \bfS_k &= \bfH_{k} \bfP_{k\mid k-1} \bfH_{k}^T + \hbfN_k,\\
            \bfK_k &= \bfP_{k\mid k-1} \bfH_{k}^T \bfS_k^{-1},\\
            \hbochi_{k\mid k} &= \hbochi_{k\mid k-1} \exp_G(\bfK_k \bfz_k),\\
            \bfP_{k\mid k} &= (\bfI - \bfK_k \bfH_{k}) \bfP_{k\mid k-1},
        \end{aligned}\right.$\\
        Prop. & $\left\{\begin{aligned}
            \hbochi_{k+1\mid k} &= \bff(\hbochi_{k\mid k}, \bfu_k, \bfzero),\\
            \bfP_{k+1\mid k} &= \bfF_{k} \bfP_{k\mid k} \bfF_{k}^T + \bfG_{k} \bfQ_k \bfG_{k}^T,
        \end{aligned}\right.$
    \end{tabular}
\end{equation}
 with $\bfz_k$  the innovation (\ie,   prediction error), $\bfS_k$  the innovation covariance, and $\bfK_k$   the Kalman gain. As is customary in extended Kalman filtering, Jacobian matrices $\bfF_k,\bfG_k$ arise from first-order linearizations. However, in the case of invariant filtering, they are defined with respect to model \eqref{eq:concentrated_gaussian}, \ie,  
\begin{equation}\label{eq:F_k}
    \begin{aligned}
        &\bff(\bochi\exp_G(\boxi), \bfu, \bfw) =  \bff(\bochi, \bfu, \bfzero)\cdot\\
        &\qquad\quad\exp_G\left(\bfF\boxi +\bfG \bfw+ \mathcal{O}(\|\boxi\|^2,\|\bfw\|^2,\|\boxi\| \|\bfw\|)\right).     
    \end{aligned}
\end{equation}
When the function $\bff$ possesses the group affine property (\ie, is of the form $\bff(\bochi , \bfu, \bfw)=\bar\bff(\bochi , \bfu)\mathbf{g}( \bfw)$ where $\bar\bff$ satisfies $\bar\bff(\bfa\bfb,\bfu)=\bar\bff(\bfa, \bfu)\bar\bff(\bfI,\bfu)^{-1}\bar\bff(\bfb,\bfu)$, for all $\bfa,\bfb \in G$, $\bfu \in \reals^m$, see~\cite{barrau2018invariant}), the Jacobian $\bfF_k$ becomes independent of the current state estimate $\hbochi_{k\mid k}$ \cite{barrau2016invariant,barrau2018invariant}. Moreover, in the absence of process noise, this results in the exact error dynamics $\boxi_{k+1\mid k} = \bfF_k \boxi_{k\mid k}$, where the Jacobian $\bfF_k$, i.e., the first order, fully captures the nonlinearity of $\bff$. This property called log-linearity is a key feature of invariant filtering. 
Note that an \IEKF{} can   be devised even if $\bff$ is not group affine, based on     expansion \eqref{eq:F_k}, and the theory to follow applies.
 
 Similarly, we define $\bfH_k$ as the Jacobian of the innovation $\bfz_k$ w.r.t.\ $\boxi_{k\mid k-1}$. To compute it, we first express $\bfy_k$ as $\hbochi_{k\mid k-1} \exp_G(\boxi_{k\mid k-1})\bfd_k + \bfn_k$. We define the innovation   $\bfz_k = \hbochi_{k\mid k-1}^{-1} \bfy_k - \bfd_k=\exp_G(\boxi_{k\mid k-1})\bfd_k - \bfd_k +  \hbochi_{k\mid k-1}^{-1} \bfn_k$, and perform a   first-order expansion w.r.t.\  $\boxi_{k\mid k-1}$. The corresponding Jacobian $ \bfH_k$ is then defined through
\begin{equation}\label{eq:H_k}
    \exp_G(\boxi) \bfd_k = \bfd_k + \bfH_k \boxi + \mathcal{O}(\|\boxi\|^2),
\end{equation} 
whose detailed expression is to be found at Equation \eqref{eq:jac_formula}. 
Due to the specific form of the output and the definition of the innovation, the Jacobian $\bfH_k$ depends only on $\bfd_k$ and is thus independent of $\hbochi_{k\mid k-1}$, a feature of invariant filtering.

\subsection{The \IIEKF{} equations} 
When receiving a noisy measurement in the left-invariant form~\eqref{inv:measurements}, the \IIEKF{} aims to find the maximum \textit{a posteriori} (MAP) estimate, described in the following result.
\begin{proposition}
    Starting from the prior~\eqref{eq:concentrated_gaussian}, with $l=k-1$, supposed to encode the state distribution at time $k$ conditional on past information  $\bfy_0,\ldots,\bfy_{k-1}$, the MAP estimate in the light of latest measurement $\bfy_k$ is given by $\hbochi_{k\mid k}^\star = \hbochi_{k\mid k-1}\exp_G(\boxi^\star)$, where $\boxi^\star$ solves the optimization problem
    \begin{equation} \label{eq:IEKF_max_posterior}
        \boxi^\star = \argmin\limits_{\boxi} \frac{1}{2}\|\boxi\|^2_{\bfP_{k\mid k-1}} + \frac{1}{2}\|\bfz_{k} - \exp_G(\boxi)\bfd_k + \bfd_k\|^2_{\hbfN_k},
    \end{equation}
    where $\bfz_k$ is the innovation from~\eqref{eq:IEKF_equations},  $\|\boeta\|^2_{\Xi}:=\boeta^T\Xi^{-1}\boeta,$ and where matrices $\bfP_{k\mid k-1}$ and $\hbfN_k$ are assumed to be invertible.
\end{proposition}
\begin{proof}
    From Bayes' rule $p(\boxi_{k\mid k-1}\mid\bfy_k) = p(\boxi_{k\mid k-1}\mid\bfz_k) \propto p(\boxi_{k\mid k-1})p(\bfz_k\mid\boxi_{k\mid k-1})$, 
    with all densities $p$ implicitly conditional on past information $\bfy_0,\ldots,\bfy_{k-1}$. Thus $p(\boxi_{k\mid k-1}\mid\bfy_k)\propto\exp(-\frac{1}{2}\|\boxi_{k\mid k-1}\|^2_{\bfP_{k\mid k-1}})\exp(-\frac{1}{2}\|\bfz_{k} - \exp_G(\boxi_{k\mid k-1})\bfd_k + \bfd_k\|^2_{\hbfN_k})$, maximized by $\boxi^\star$.
\end{proof}

The optimization problem~\eqref{eq:IEKF_max_posterior} does not admit a closed-form solution. The \IEKF{} addresses this by approximating $\exp_G(\boxi)$ using its first-order Taylor expansion around $\bfzero$. Inspired by the iterated \EKF{}~\cite{bell1993iterated}, a more accurate solution can be obtained by applying the GN algorithm to iteratively refine the estimate until convergence. This process results in a sequence of updates that closely resemble those of the \IEKF{}. The resulting Iterated \IEKF{} (\IIEKF{}) is outlined in Algorithm~\ref{alg:IIEKF}. Details are provided in Appendix~\ref{app:IIEKF}.

\begin{algorithm}
\caption{The iterated invariant extended Kalman filter (\IIEKF{})}
\begin{algorithmic}[1]
\STATE Choose the initial state $\hbochi_{0\mid 0}\in G$ and initial covariance $\bfP_{0\mid 0} = \mathrm{Cov}(\boxi_{0\mid 0}) \in \reals^{n\times n}$.
\LOOP
\STATE Define   Jacobians $\bfF_k$, $\bfG_k$,  $\bfH_k$ of \IEKF{} framework.
\STATE Define  noise covariances $\bfQ_k,\bfN_k$.
\STATE
\STATE\COMMENT{Update}
\STATE $\mathrlap{\bfz_k}\hphantom{\boxi_{k\mid k-1}^{i}} \gets \hbochi_{k\mid k-1}^{-1}\bfy_k - \bfd_k$
\STATE $\mathrlap{\hbfN_k}\hphantom{\boxi_{k\mid k-1}^{i}} \gets \hbochi_{k\mid k-1}^{-1}\bfN_k(\hbochi_{k\mid k-1}^{-1})^T$
\STATE $\mathrlap{\boxi_{k\mid k-1}^{i}}\hphantom{\boxi_{k\mid k-1}^{i}} \gets \bfzero$
\WHILE {$\boxi_{k\mid k-1}^{i}$ not converged}
\STATE $\mathrlap{\bfH_k^{i}}\hphantom{\boxi_{k\mid k-1}^{i}} \gets \exp_G(\boxi_{k\mid k-1}^{i})\bfH_k\calJ_r(\boxi_{k\mid k-1}^{i})$
\STATE $\mathrlap{\bfS_k^{i}}\hphantom{\boxi_{k\mid k-1}^{i}} \gets \bfH_k^{i}\bfP_{k\mid k-1}(\bfH_k^{i})^T + \hbfN_k$
\STATE $\mathrlap{\bfK_k^{i}}\hphantom{\boxi_{k\mid k-1}^{i}} \gets \bfP_{k\mid k-1}(\bfH_k^{i})^T(\bfS_k^i)^{-1}$
\STATE $\mathrlap{\bfz_k^{i}}\hphantom{\boxi_{k\mid k-1}^{i}} \gets \bfz_k - \exp_G(\boxi_{k\mid k-1}^{i})\bfd_k + \bfd_k + \bfH_k^{i}\boxi_{k\mid k-1}^{i}$
\STATE $\mathrlap{\boxi_{k\mid k-1}^{i}}\hphantom{\boxi_{k\mid k-1}^{i}} \gets \bfK_k^{i} \bfz_k^{i}$
\ENDWHILE
\STATE $\mathrlap{\hbochi_{k\mid k}}\hphantom{\bfP_{k\mid k}} \gets \hbochi_{k\mid k-1}\exp_G(\boxi_{k\mid k-1}^{i})$
\STATE $\mathrlap{\bfK_k}\hphantom{\bfP_{k\mid k}} \gets \bfP_{k\mid k-1}\bfH_k^{T}(\bfH_k\bfP_{k\mid k-1}\bfH_k^{T} + \hbfN_k)^{-1}$
\STATE $\mathrlap{\bfP_{k\mid k}}\hphantom{\bfP_{k\mid k}} \gets (\bfI - \bfK_k\bfH_k)\bfP_{k\mid k-1}$
\STATE
\STATE \COMMENT{Propagation}
\STATE $\mathrlap{\hbochi_{k+1\mid k}}\hphantom{\hbochi_{k+1\mid k}} \gets \bff(\hbochi_{k\mid k}, \bfu_k, \bfzero)$
\STATE $\mathrlap{\bfP_{k+1\mid k}}\hphantom{\hbochi_{k+1\mid k}} \gets \bfF_k\bfP_{k\mid k}\bfF_k^T + \bfG_k\bfQ_k\bfG_k^T$
\ENDLOOP
\end{algorithmic}
$\calJ_r(\boxi)$ denotes the right Jacobian~\cite{chirikjian2009stochastic} of $G$ at $\boxi$, see \eqref{eq:right_jac}.
\label{alg:IIEKF}
\end{algorithm}
Iterating proves particularly useful when measurement noise is low, as we will see in the sequel. However, in the limit case where the measurement noise magnitude approaches zero, the innovation covariance $\bfS_k$ may become rank-deficient \cite{goffin2023invariant}, in which case the Kalman gain is undefined. An adapted version of the \IIEKF{} for this specific scenario is detailed in Algorithm~\ref{alg:nf_IIEKF}. Details are provided in Appendix~\ref{app:nf_IIEKF}.

\begin{algorithm}
\caption{The \IIEKF{} with noise-free measurements}
\begin{algorithmic}
\STATE Same as Algorithm~\ref{alg:IIEKF} with the following exceptions.
\STATE $\,$~5: Find $\bfL_{k\mid k-1}$ such that $\bfP_{k\mid k-1} = \bfL_{k\mid k-1}\bfL_{k\mid k-1}^T$.
\STATE 13: $\mathrlap{\bfK_k^{i}}\hphantom{\bfK_k^{i}} \gets \bfL_{k\mid k-1}(\bfH_k^i\bfL_{k\mid k-1})^\dagger$
\STATE 18: $\mathrlap{\bfK_k}\hphantom{\bfK_k^{i}} \gets \bfL_{k\mid k-1}(\bfH_k\bfL_{k\mid k-1})^\dagger$
\end{algorithmic}
Here, $(\cdot)^\dagger$ denotes the Moore-Penrose pseudo-inverse.
\label{alg:nf_IIEKF}
\end{algorithm}

\section{Theoretical properties of the \IIEKF}\label{sec:properties}
Deriving properties in the presence of measurement noise seems out of reach, as is often the case in nonlinear filtering (\eg, the log-linear error property of~\cite{barrau2016invariant} only holds in the absence of process noise). However, we show   that the \IIEKF{} possesses strong properties in the absence of measurement noise (whereas the state estimate and the propagation may be noisy). We therefore consider noise-free measurement 
\begin{equation}\label{eq:nf_meas}
    \bfy_k = \bochi_k \bfd_k.
\end{equation}
In this case, a measurement defines a subset of the state space.
\begin{definition}\label{def:observed_set}
    We call \emph{observed set} associated with the noise-free measurement $\bfy_k = \bochi_k \bfd_k$ the subset 
    \begin{equation}\label{observed_set}
        \obset{\bochi \bfd_k = \bfy_k} := \{\bochi \in G \mid \bochi \bfd_k = \bfy_k\}.
    \end{equation}
\end{definition}Definition~\ref{def:compatibility} establishes criteria for  local compatibility of a Gaussian filter on Lie groups with   measurement  \eqref{eq:nf_meas}. 
\begin{definition}\label{def:compatibility}
    The estimate $(\hbochi_{k\mid l}, \bfP_{k\mid l})$ is said to be \emph{locally compatible with noise-free measurement $\bfy_k = \bochi_k \bfd_k$} if 
    \begin{enumerate}
        \item $\hbochi_{k\mid l}\in \obset{\bochi \bfd_k= \bfy_k}$,
        \item $\bfH_{k} \bfP_{k\mid l} \bfH_{k}^{T}=\bfzero$,
    \end{enumerate}
    where $\bfH_{k}$ denotes the Jacobian defined in~\eqref{eq:H_k}.
\end{definition}
Since these criteria are not immediately obvious, we briefly motivate their formulation. Using model \eqref{eq:concentrated_gaussian} and linearizing around $\boxi_{k\mid l}=\bfzero$ via \eqref{eq:H_k}, we obtain $\bfy_k = \bochi_k \bfd_k = \hbochi_{k\mid l}\exp_G(\boxi_{k\mid l})\bfd_k \approx \hbochi_{k\mid l} \bfd_k + \hbochi_{k\mid l}\bfH_k \boxi_{k\mid l}  $. As $\boxi_{k\mid l} \sim \mathcal{N}(\bfzero, \bfP_{k\mid l})$, for this  linear approximation of \eqref{eq:nf_meas} to hold (almost surely) we must have both $\hbochi_{k\mid l} \bfd_k = \bfy_k$ (letting $\boxi_{k\mid l}=0$), and  $\bfH_k \boxi_{k\mid l} = \bfzero$ whenever $\boxi_{k\mid l}$ in $\im \bfP_{k\mid l}$, that is, 1) and 2). 

Before turning to the \IIEKF{} and its properties in this context, let us further analyze the specific structure of the measurements under consideration.

\subsection{Properties of considered noise-free measurements }

The theory of (left) invariant filtering focuses on measurements of the form~\eqref{inv:measurements}, see~\cite{barrau2016invariant}, which boil down to \eqref{eq:nf_meas} when the noise is turned off. It turns out  that the corresponding "observed set"~\eqref{observed_set} possesses an interesting structure. 
\begin{proposition}\label{prop:structure}
    Let $\bfH_{k}$ denote the Jacobian from invariant filtering defined in~\eqref{eq:H_k}. We have the following properties
   \begin{enumerate}
          \item $\bfH_{k}\boxi=\bfzero,~\boxi\in\reals^n \Rightarrow \exp_G\left(\boxi\right)\bfd_k=\bfd_k$,
       \item   $\liealg{s}:=\{ \Lmap{\boxi}\in \liealg{g} \mid \bfH_{k}\boxi = \bfzero\}$ is a Lie subalgebra of $\liealg{g}$.
   \end{enumerate}
\end{proposition}
\begin{proof}
    As we are dealing with matrix Lie groups, the exponential coincides with the matrix exponential $\exp_m$ as follows  
    \begin{equation}\label{eq:expo}
        \exp_G(\boxi)= \exp_m(\Lmap{\boxi}) = \bfI + \sum\limits_{l=1}^{+\infty} \frac{\Lmap{\boxi}^l}{l!}.
    \end{equation}
    Keeping only first-order terms proves, in passing, that the Jacobian $\bfH_k$ from~\eqref{eq:H_k} writes
    \begin{equation}\label{eq:jac_formula}
        \bfH_k\boxi = \Lmap{\boxi}\bfd_k. 
    \end{equation} 
   Thus $\bfH_k\boxi = \bfzero\Rightarrow \Lmap{\boxi}^l\bfd_k=\bfzero$ by induction, and then
   \begin{equation}
       \exp_G(\boxi) \bfd_k = \bfd_k + \sum\limits_{l=1}^{+\infty} \frac{\Lmap{\boxi}^l \bfd_k}{l!} = \bfd_k,
   \end{equation}
   proving the first point. Regarding the second point, let $\bodelta, \bozeta \in \reals^n$ be such that $\Lmap{\bodelta}, \Lmap{\bozeta} \in \liealg{s}$. Considering the standard bilinear skew-symmetric Lie bracket defined as $[\bfA,\bfB]= \bfA\bfB-\bfB\bfA$, we have $[\Lmap{\bodelta}, \Lmap{\bozeta}]\bfd_k = \Lmap{\bodelta}\Lmap{\bozeta}\bfd_k - \Lmap{\bozeta}\Lmap{\bodelta}\bfd_k = \bfzero$ by the definition of $\liealg{s}$ and~\eqref{eq:jac_formula}. This proves that $[\Lmap{\bodelta}, \Lmap{\bozeta}] \in \liealg{s}$, confirming that $\liealg{s}$ is closed under the Lie bracket and is therefore a subalgebra of $\liealg{g}$. 
\end{proof}
The first point of Proposition~\ref{prop:structure} is directly useful for proving the compatibility properties of the \IIEKF{}, in the sense of Definition \ref{def:compatibility}. This will be established in the next subsection, relying on the following key Lemma~\ref{lem:confined_subspace}. We also note, in passing, that the second point will only become relevant when extending the present theory to the iterated \EKF{} on Lie groups (\LGIEKF{}) of~\cite{bourmaud2015continuous}, under certain conditions.

\begin{lemma} \label{lem:confined_subspace}
    Let $\bfH_{k}$ denote the Jacobian from invariant filtering defined in~\eqref{eq:H_k}. We have
    \begin{equation}
        \begin{array}{c}
              \hbochi_{k\mid l} \in\obset{\bochi\bfd_k=\bfy_k}\quad\text{and}\quad\bfH_{k}\boxi=\bfzero,~\boxi\in\reals^n\\
              \Downarrow\\
              \hbochi_{k\mid l}\exp_G\left(\boxi\right)\in \obset{\bochi\bfd_k=\bfy_k}.
        \end{array}
    \end{equation}
\end{lemma}
\begin{proof}
   Using the first point of  Proposition \ref{prop:structure} yields
    \begin{equation}
       \bfH_k\boxi=\bfzero\Rightarrow \hbochi_{k\mid l} \exp_G(\boxi) \bfd_k = \hbochi_{k\mid l}\bfd_k=\bfy_k,
    \end{equation}
    which proves $\hbochi_{k\mid l} \exp_G(\boxi) \in \obset{\bochi\bfd_k=\bfy_k}$.
\end{proof}
 
 A direct and important consequence is  the following.
 \begin{proposition}
     Assume the two points of Definition~\ref{def:compatibility} are satisfied. Then, the entire probability distribution encoded by~\eqref{eq:concentrated_gaussian} is in fact (almost surely) contained within the observed set $\obset{\bochi\bfd_k=\bfy_k}$, \ie, the estimate is \emph{globally compatible with the noise-free measurement $\bfy_k = \bochi_k\bfd_k$}. 
 \end{proposition}
In other words, if the filter manages to encode the compatibility assumptions correctly locally, it will convey an estimated distribution being wholly consistent with  available information, beyond the first order. This    shall play a key role.

\subsection{Compatibility of the \IIEKF{} with the measurements}\label{propert:subsec}

Let us first focus on the first point of Definition~\ref{def:compatibility}. As the magnitude of the measurement noise tends to zero, the optimization problem~\eqref{eq:IEKF_max_posterior} boils down to finding the smallest~$\boxi$ in the sense of the metric induced by $\bfP_k$, that satisfies the hard constraint $\|\bfz_{k} - \exp_G(\boxi)\bfd_k + \bfd_k\|^2=0$. As $\bfz_k=\hbochi_{k\mid k-1}^{-1}\bfy_k - \bfd_k$, this constraint boils down to $\hbochi_{k\mid k-1}^{-1}\bfy_k=\exp_G\left(\boxi\right)\bfd_k$ and the minimizer $\boxi^\star$ we are seeking necessarily satisfies $\hbochi_{k\mid k-1}\exp_G\left(\boxi^\star\right)\bfd_k= \bfy_k$. Provided that the \IIEKF{} iterate $\boxi_{k \mid k-1}^i$, which is a GN descent, indeed converges to such a minimizer, the updated state $\hbochi_{k\mid k} = \hbochi_{k \mid k-1} \exp_G(\boxi^\star)$  will (by definition) belong  to the observed set, \ie, $\hbochi_{k \mid k} \in \obset{\bochi\bfd_k = \bfy_k}$, ensuring the first point of Definition~\ref{def:compatibility}. As the optimization problem is not convex,  convergence of the GN descent is only guaranteed if the true error $\boxi_{k\mid k-1}$ is sufficiently close to $\bfzero$. However, by contrast, non-iterative update schemes, such as the \IEKF{} update, do not guarantee that $\hbochi_{k\mid k} \in \obset{\bochi\bfd_k = \bfy_k}$ even for small initial errors. In the upcoming simulations of Section~\ref{sec:crane}, and more generally in all simulations we have performed, the \IIEKF{} always converges to the observed set. 

Let us now turn to the second point of Definition~\ref{def:compatibility}.
\begin{theorem} \label{thm:ric}
    The \IIEKF{} updated covariance matrix $\bfP_{k\mid k}$ in the light of  noise-free measurement  $\bfy_k = \bochi_k\bfd_k$  naturally ensures  $\bfH_k\bfP_{k\mid k}\bfH_k^T = \bfzero$.
\end{theorem}
\begin{proof} 
    This is a consequence of linear Kalman theory, which may be proved as follows. When facing noise-free measurements, the Kalman gain is computed as $\bfK_k=\bfL_{k\mid k-1}(\bfH_k\bfL_{k\mid k-1})^\dagger$, where $\bfP_{k\mid k-1} = \bfL_{k\mid k-1}\bfL_{k\mid k-1}^T$. Letting $\bfA_k=\bfH_k\bfL_{k\mid k-1}$, we have: $            \bfH_k\bfP_{k\mid k}\bfH_k^T 
             = \bfH_k(\bfI - \bfK_k\bfH_k)\bfP_{k\mid k-1}\bfH_k^T, 
             = (\bfA_k - \bfA_k(\bfA_k)^\dagger \bfA_k)\bfA_k^T = \bfzero,
     $   
    where we used $\bfA_k\bfA_k^\dagger\bfA_k = \bfA_k$.
\end{proof}

\subsection{Consequences of the compatibility property}
When fed with a measurement~\eqref{eq:nf_meas} (resp.\ with a measurement~\eqref{inv:measurements} with small noise), the information that the state lies in (resp.\ is close to)  the observed subset $\obset{\bochi \bfd_k= \bfy_k}$ should be well encoded in the filter, that is, no immediate subsequent measurement should be able to destroy that piece of information.  
The following result shows that the \IIEKF{} indeed can inherently "lock in" perfect information, which is akin to the behavior of the (linear) Kalman filter when confronted with a noise-free linear measurement. 
\begin{theorem}
    \label{thm:update_compatibility}
    Let $\bbfy_k = \bochi_k \bbfd_k$ represent a noise-free piece of information about the true state. Consider an \IIEKF{} whose current distribution is compatible with this information in the sense of Definition~\ref{def:compatibility}, meaning that $\hbochi_{k\mid k-1} \in \obset{\bochi\bbfd_k=\bbfy_k}$ and that $\bbfH_k\bfP_{k\mid k-1}\bbfH_k^T = \bfzero$, where $\bbfH_k$ is the Jacobian associated with $\bbfd_k$. If the estimate is subsequently updated using a (possibly noisy) measurement $\bfy_k=\bochi_k\bfd_k+\bfn_k$, then the corresponding updated estimate $(\hbochi_{k\mid k}, \bfP_{k\mid k})$, which incorporates this new measurement $\bfy_k$, remains compatible with the deterministic information $\bbfy_k=\bochi_k \bbfd_k$. Mathematically,
    \begin{equation}
        \begin{array}{c}
              \hbochi_{k\mid k-1} \in \obset{\bochi\bbfd_k=\bbfy_k} \quad \text{and} \quad \bbfH_k\bfP_{k\mid k-1}\bbfH_k^T = \bfzero, \\
                \Downarrow \\
               \hbochi_{k\mid k} \in \obset{\bochi\bbfd_k=\bbfy_k} \quad \text{and} \quad \bbfH_k\bfP_{k\mid k}\bbfH_k^T = \bfzero.
        \end{array}\label{eq:geometry_respect}
    \end{equation}
\end{theorem}
\begin{proof}
    Let us first consider the first point of Definition~\ref{def:compatibility}. During the update with $\bfy_k$, we see at line 17 of Algorithm~\ref{alg:IIEKF} that $\hbochi_{k\mid k} = \hbochi_{k\mid k-1}\exp_G(\boxi_{k \mid k-1}^i)$, where $\boxi_{k \mid k-1}^i = \bfK_k^{i}\bfz_k^{i} \in \im \bfK_k^{i} \subseteq \im \bfP_{k\mid k-1}$, with $\im \bfK_k^{i}$ denoting the span of $\bfK_k^{i}$. Since $\bbfH_k\bfP_{k\mid k-1}\bbfH_k^T=\bfzero$ by assumption, we necessarily have $\bbfH_k \boxi_{k \mid k-1}^i=\bfzero$. Then, as $\hbochi_{k\mid k-1} \in \obset{\bochi_k \bbfd_k = \bbfy_k}$, a direct application of Lemma~\ref{lem:confined_subspace} ensures $\hbochi_{k\mid k} \in \obset{\bochi \bbfd_k = \bbfy_k}$.
    
    Now, regarding the second point of Definition~\ref{def:compatibility}, the condition $\bbfH_k\bfP_{k\mid k-1}\bbfH_k^T = \bfzero$ implies $\bfP_{k\mid k-1}\bbfH_k^T = \bfzero$ by symmetry. The update of the \IIEKF{} with measurement $\bfy_k$, see line 19 of Algorithm~\ref{alg:IIEKF}, gives $\bfP_{k\mid k} = (\bfI - \bfK_k\bfH_k)\bfP_{k\mid k-1}$, where $\bfH_k$ is the Jacobian associated with $\bfd_k$. Since $\im \bfK_k \subseteq \im \bfP_{k\mid k-1}$, the updated covariance $\bfP_{k\mid k}$ also verifies $\bfP_{k\mid k}\bbfH_k^T = \bfzero$, whatever $\bfH_k$. Consequently, $\bbfH_k\bfP_{k\mid k}\bbfH_k^T=\bfzero$.
\end{proof}

Once the noise-free information has been properly accounted for, the \IIEKF{} effectively "focuses" on a problem of reduced dimensionality: the subsequent update correctly adjusts the state within   the observed set $\obset{\bochi\bbfd_k=\bbfy_k}$. This significantly improves the filter efficiency in practice, as will be demonstrated in the simulations. By contrast, a classical \EKF{} or even an iterated \EKF{} may step out of the appropriate subset,   resulting in less efficient estimation updates.

\paragraph*{Illustration on $\SO{3}$} 
Albeit a simple example, it is pedagogical to illustrate the results on a problem involving rotation matrices. Let $\bochi_k \in \SO{3}$ be an unknown rotation matrix representing the orientation of a drone in an inertial frame. Suppose an \IIEKF{} is used to estimate $\bochi_k$, and that prior observation has enabled perfect knowledge of the gravity vector in the drone body frame $\bbfd_k$. As a result, the current estimate $(\hbochi_{k\mid k-1}, \bfP_{k\mid k-1})$ is compatible, in the sense of Definition~\ref{def:compatibility}, with the noise-free information $\bfg = \bochi_k \bbfd_k$. Now, suppose a new measurement $\bfy_k = \bochi_k\bfd_k + \bfn_k$ is received from another sensor. When updating the estimate using $\bfy_k$, all iterates of the \IIEKF{}, denoted by $\hbochi_{k\mid k-1}\exp_{\SO{3}}(\boxi_{k \mid k-1}^i)$, remain within the observed set $\obset{\bochi \bbfd_k = \bfg}$. Moreover, the update also ensures  $\bbfH_k \bfP_{k \mid k} \bbfH_k^T = \bfzero$. Consequently, once the noise-free information $\bfg = \bochi_k \bbfd_k$ is incorporated, the \IIEKF{} operates in a subset of reduced dimensionality—namely, it seeks planar rotations within a 2D subspace. This makes the filter significantly more efficient because 1) it performs updates directly within a reduced subspace that contains the true state, and 2) it  wholly  preserves the information from the previous measurements. Note that, if noise is low instead, the latter remains approximately true---as there is a continuum between the noise-free and low-noise cases---which should make the filter very efficient in practice when measurement noise is low. 

\subsection{Extending the theory  to the iterated \EKF{} on Lie groups}
\label{subsec:distinction_bourmaud}

The \LGIEKF{}, introduced in~\cite{bourmaud2016intrinsic}, transposes the standard iterated \EKF{} to Lie group state spaces, by leveraging an intrinsic GN method on the group, that accounts for its manifold structure, notably curvature.   \LGIEKF{} of \cite{bourmaud2016intrinsic} is a general algorithm, meant to address general problems on Lie groups having no specific structure, beyond the intrinsic nature of the state space. By contrast, our \IIEKF{} addresses specific types of observations, namely \eqref{inv:measurements} (or its right-invariant counterpart $\bfy_{k}  = \bochi_k^{-1}\bfd_k + \bfn_{k}$, in which case an iterated \emph{right-invariant} EKF  should be used, by mimicking the difference between the LIEKF and the RIEKF of \cite{barrau2016invariant}). 
It defines a modified innovation     $\bfz_k=\hbochi_k^{-1}\bfy_k - \bfd_k$, allowing for the use  of the invariant Kalman filter framework of \cite{barrau2016invariant,barrau2018invariant}.  The standard invariant EKF has had various successes across control, navigation, and robotics, and is clearly a different algorithm than the \LGEKF{} of \cite{bourmaud2013discrete}.

That said, the questions one could address are as follows. Assume we use the theory of invariant filtering to endow the \LGIEKF{}   with all the features of the \IIEKF{}—namely, measurements of the form~\eqref{eq:nf_meas}, uncertainty model~\eqref{eq:concentrated_gaussian}, and the Kalman gain from Algorithm~\ref{alg:nf_IIEKF} to address  rank deficiency issues, as well as the appropriate choice regarding left or right versions (note that \cite{bourmaud2016intrinsic} chooses a systematic right-invariant error based setting whereas we advocate to choose the error depending on the form of the observations), resulting in an "invariant" version of the \LGIEKF{}, that we call \emph{adapted \LGIEKF{}}, see Algorithm~\ref{alg:LG-IEKF}. Then,  1) do we actually recover the \IIEKF{}, and 2) if no, does it inherit  any of the properties we have proved, given that no such properties have ever been proved for the  \LGIEKF{}? The answer to point 1  is no, as intrinsic GN on the group entails adding various Jacobians    and  differs  from the GN algorithm addressing  \eqref{eq:IEKF_max_posterior}, that directly optimizes  in a linear space. The answer to point 2  is yes,  the properties above do carry over. This is detailed,   proved, and discussed, in Appendix \ref{Appendix:Bourmaud}.

\begin{algorithm}
\caption{The adapted \LGIEKF{} from~\cite{bourmaud2016intrinsic} when measurements   \eqref{eq:nf_meas} are considered, uncertainty model \eqref{eq:concentrated_gaussian} is used,  potential rank deficiency issues in the Kalman gain are addressed via the noise-free gain from Algorithm~\ref{alg:nf_IIEKF}, and a left-invariant error is used instead. See \eqref{eq:right_jac} for definition of $\calJ_r$.}
\begin{algorithmic}
\STATE Same as Algorithm~\ref{alg:nf_IIEKF} with the following exceptions.
\STATE 11: $\mathrlap{\bfH_k^{i}}\hphantom{\bfP_{k\mid k}} \gets \hbochi_{k\mid k-1}\exp_G(\boxi_{k \mid k-1}^i)\bfH_k\calJ_r(\boxi_{k \mid k-1}^i)$
\STATE 14: $\mathrlap{\bfz_k^{i}}\hphantom{\bfP_{k\mid k}} \gets \bfy_k - \hbochi_{k\mid k-1}\exp_G(\boxi_{k \mid k-1}^i) \bfd_k + \bfH_k^i\boxi_{k \mid k-1}^i$
\STATE 19: $\mathrlap{\bfP_{k\mid k}}\hphantom{\bfP_{k\mid k}} \gets \calJ_r(\boxi_{k \mid k-1}^i)(\bfI - \bfK_k^i\bfH_k^i)\bfP_{k\mid k-1} \calJ_r(\boxi_{k \mid k-1}^i)^T$
\end{algorithmic}
\label{alg:LG-IEKF}
\end{algorithm}

\section{Application of engineering interest}\label{sec:crane}
We propose using the \IIEKF{} to estimate the position, velocity, and orientation (extended pose) of the hook of a crane transporting a load, equipped with an IMU. Mounting an IMU on a crane hook for real-time sensor data transmission is technically very feasible nowadays   \cite{rauscher2018motion}, and may open the door to new automation capabilities since it allows for feedback. 
 Leveraging the Lie group $\SE{3}{2}$ introduced by invariant filtering theory in \cite{barrau2016invariant}, we represent the state as
\begin{equation}
    \bochi_k = \begin{bmatrix}
        \bfR_k & \bfv_k & \bfp_k\\
        \bfzero & 1 & 0\\
        \bfzero & 0 & 1
    \end{bmatrix} \in \SE{3}{2},\label{chi:def:eq}
\end{equation}
where $\bfR_k \in \SO{3}$ is the rotation matrix between the IMU and inertial frames, and $\bfv_k,\,\bfp_k \in \reals^3$ are the IMU velocity and position in the inertial frame. The inertial frame is fixed at the crane cable attachment point with the $z$-axis oriented upward, and remains stationary. The hook and IMU frames are assumed perfectly aligned. See our preliminary conference paper~\cite{goffin2023invariant} for more details on this application.  

Neglecting IMU biases, the system dynamics write
\begin{subequations} \label{eq:sys_dyn}
\begin{align}
    \bfR_{k+1} &
    = \bfR_k\exp_{\SO{3}}((\boomega_k + \bfw^{\omega}_k) dt),\label{eq:rot_dyn}\\
    \bfv_{k+1} &= \bfv_k + \left(\bfR_k (\bfa_k + \bfw^{a}_k) + \bfg\right)\,dt,\label{eq:vel_dyn}\\
    \bfp_{k+1} &= \bfp_k + \bfv_k\,dt,\label{eq:pos_dyn}
\end{align}
\end{subequations}
where $\boomega_k$ and $\bfa_k$ denote the IMU angular velocity and linear acceleration, $\bfw^{\omega}_k$ and $\bfw^{a}_k$ represent Gaussian noise in the gyroscope and accelerometer, and $\bfg$ is the gravity vector.  The IMU is attached to a cable whose length $L_k$  is very accurately measured in modern cranes by motor encoders. We hence use it as noise-free measurement. Mathematically, this writes $\bochi_k \bfd_k = \bfzero$, with $\bfd_k = \begin{bmatrix} 0 & 0 & L_k & 0 & 1 \end{bmatrix}^T$. The last two rows of $\bochi_k\bfd_k$ are discarded as they do not contain  information. 

A short simulation is carried out. The hook starts with an initial orientation of $\SI{-45}{\degree}$ around the $y$-axis, zero velocity, and position $\bfp_0 = \begin{bmatrix}\sqrt{2}/2 & 0 & -\sqrt{2}/2\end{bmatrix}^T$. The filter and IMU operate at $\SI{100}{\hertz}$, with a time step $dt$ of $\SI{0.01}{\second}$. The cable length follows $L_k = L_c(k\cdot dt)$, where $L_c(t)$ evolves as $\ddot L_c(t) + 12 \dot L_c(t) + 16 L_c(t) = 64$, with initial conditions $L_c(0)=\SI{1}{\meter}$ and $\dot L_c(0)=\SI{0}{\meter/\second}$.

We compare five filters: \EKF{}, \IterEKF{}, (adapted) \LGIEKF{} (implemented as in Algorithm~\ref{alg:LG-IEKF}), \IEKF{}, and \IIEKF{}. Results are averaged over 500 runs, each with the same ground truth trajectory and a random initial error drawn from~\eqref{eq:concentrated_gaussian}. The estimation is constrained to the $xz$-plane to ensure observability. Gyroscope noise has a standard deviation of $\SI{0.974}{\degree/\second}$ about the $y$-axis and $\SI{0}{\degree/\second}$ about the $x$- and $z$-axes. Accelerometer noise has a standard deviation of $\SI{0.1}{\meter/\second^2}$ along the $x$- and $z$-axes and $\SI{0}{\meter/\second^2}$ along the $y$-axis. The initial error covariance matrix encodes a standard deviation of $\SI{45}{\degree}$ for rotation about the $y$-axis and $\SI{5}{\meter/\second}$ and $\SI{5}{\meter}$ for velocity and position along the $x$- and $z$-axes. No measurement noise is assumed. However, although 
    the noise-free gain formulation used in Algorithm~\ref{alg:nf_IIEKF} is always well-defined and serves as a  basis for the theoretical analysis, it is not well-suited to practical implementations, as it can be difficult to distinguish small singular values from actual zeros in the calculation of the pseudo-inverse. As a result, the noise-free gain formulation from Algorithm~\ref{alg:nf_IIEKF} is replaced with Algorithm \ref{alg:IIEKF}, setting the measurement noise covariance to a tiny value, namely $\bfN_k = 10^{-5} \bfI$. For iterative filters, the GN method stops when $\|\boxi_{k\mid k-1}^i\|_2$ changes by less than $10^{-5}$ between iterations or reaches a maximum of 50 iterations.

Figure~\ref{fig:res} shows the mean estimation error norm over time, averaged across 500 runs, with standard deviation.  Table~\ref{tab:res} summarizes the global Root Mean Square Error (RMSE) computed over the first 15 time steps of the 500 runs. The \IIEKF{} and \LGIEKF{} exhibit similar behavior, consistently outperforming all other filters. The first two updates significantly reduce the error on average, after which subsequent updates continue to decrease the error, albeit at a slower rate. Notably, both filters successfully converge across all 500 runs, regardless of the initial error. The \IIEKF{} performs slightly better on average proving that considering \eqref{eq:IEKF_max_posterior} instead of actual GN on the group does not pose problems.

The \EKF{} and \IterEKF{} exhibit similar performance, while the \IterEKF{} is slightly better. They are outperformed by the introduced algorithm because their updates make them step out of the observed set.   The \IEKF{} surprisingly  fails to converge in most cases. This behavior is expected with low measurement noise: in a noise-free setting, the \IEKF{} update fails to satisfy the first condition in Definition~\ref{def:compatibility}, while still enforcing the second. Consequently, the estimate is confined to a subset that excludes the true state and cannot escape. This phenomenon underscores the necessity of an iterative update in the invariant filtering framework when handling accurate measurements. 


As filtering is typically performed in real time, it is crucial to assess the number of iterations required during the update stage for the \IterEKF{}, \LGIEKF{}, and \IIEKF{}. Figure~\ref{fig:histo} presents a histogram of the iteration count for these iterative filters. The results indicate that in over $80\%$ of cases, both the \IIEKF{} and \LGIEKF{} complete the update stage in just two iterations, making their computational cost relatively low. In contrast, the \IterEKF{} generally requires a higher number of iterations on average, which could lead to increased computational expense.

\begin{figure}
    \centering
    \includegraphics[width=\linewidth]{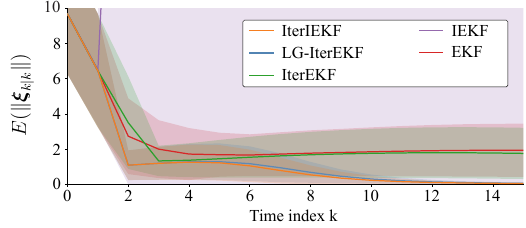}
    \caption{Evolution of the average and standard deviation of the estimation error norm.}
    \label{fig:res}
\end{figure}

\begin{table}
    \centering
    \begin{tabular}{lccccc}
        \toprule
        RMSE &    \EKF{} & \IterEKF{} & \LGIEKF{} & \IEKF{} & \IIEKF{} \\ \midrule
       Orientation & 0.805 &  0.792  &  0.592  &  1.112  &  \textbf{0.574}\\
       Velocity & 1.489 &  1.458  &  1.136  &  6.039  &  \textbf{1.116}\\
       Position & 1.054 &  1.041  &  0.871  &  1.852  &  \textbf{0.858}\\
       \bottomrule
       \rule{0pt}{0pt}
    \end{tabular}
    \caption{RMSE computed over the first 15 times steps of the 500 runs. The smallest values are highlighted in bold.}
    \label{tab:res}
\end{table}

\begin{figure}
    \centering
    \includegraphics[width=\linewidth]{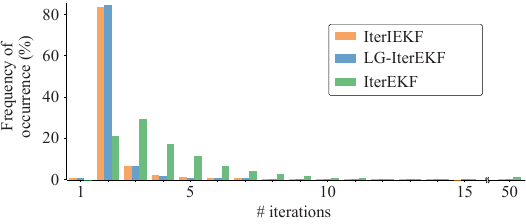}
    \caption{Histogram of the number of iterations required at the update stage of the iterated filters, computed over 500 runs.}
    \label{fig:histo}
\end{figure}

\section{Conclusion}
\label{sec:conclusion}
In this technical note, we introduced the \IIEKF{}, an iterated version of the \IEKF{} inspired by the \IterEKF{}. We analyzed its properties in the limit case of noise-free measurements, expressed in left-invariant form~\eqref{eq:nf_meas} (in case of right-invariant measurements, deriving the corresponding right-\IIEKF{} is straightforward, and the results immediately apply). Specifically, we derived criteria for assessing the local compatibility of Gaussian filters on Lie groups with this class of measurements, and leveraged their specific structure to demonstrate that these criteria also ensure global compatibility. Then, we proved   the \IIEKF{} update produces estimates that are compatible, and moreover tend to inherently "hard code" this compatibility. We also proved that the same properties hold for the \LGIEKF{} when  tailored to the invariant framework. 

We applied the \IIEKF{} to the problem of estimating the extended pose of a crane hook equipped with an IMU, where mechanical information is assimilated in the form of a noise-free pseudo-measurement. The \IIEKF{} demonstrated the best performance, closely followed by the adapted \LGIEKF{}  \cite{bourmaud2016intrinsic}, a fact explained by our present theory. The superiority of both algorithms is consistent with the theoretical guarantees we have proved, and the \IEKF{} is outperformed in this context. 

However, we recommend   using the \IIEKF{} over the \LGIEKF{}, because  the \IEKF{} possesses convergence properties (notably log-linearity of the error \cite{barrau2016invariant}) and consistency properties that have never been proved for the  \LGEKF{}. In particular, the unobservability consistency properties of the \IEKF{} \cite{barrau2018invariant,wu2017invariant,heo2018consistent} may  not be inherited by the \LGIEKF{},   although proving this fact  is left for future research.

Finally, simulations   showed   only a few iterations were needed for the GN method used in the \IIEKF{} update to converge, making the \IIEKF{} viable for real time.


As a perspective, we aim to apply this iterative estimation technique to other problems where the \IEKF{} has proven effective, particularly in high-precision navigation. Moreover, we will explore its application in state estimation   where constraints could be enforced as noise-free pseudo-measurements.

\appendices
\section{The \IIEKF{} derivation}
\label{app:IIEKF_derivation}

\subsection{Standard case: noisy measurements}\label{app:IIEKF}
Akin to the conventional iterated \EKF{}~\cite{bell1993iterated}, we intend to solve~\eqref{eq:IEKF_max_posterior} iteratively using the Gauss-Newton method.

\begin{lemma}[from~\cite{bell1993iterated}] \label{lem:Kalman_upd}
    Consider the optimization problem
    \begin{equation}
        \bfx^\star = \argmin_{\bfx}\frac{1}{2}\|\bfx - \bbfx\|^2_{\bfA} + \frac{1}{2}\|\bfb - \bfh(\bfx)\|^2_{\bfB},\label{optim:bell}
    \end{equation}
    with $\bfx, \bbfx \in \reals^n$, $\bfb \in \reals^m$, $\bfA \in GL_n(\reals)$, $\bfB \in GL_m(\reals)$ and $\bfh \in C^1(\reals^n, \reals^m)$. The Gauss-Newton method applied to~\eqref{optim:bell} yields the sequence of estimates
    \begin{equation}
        \bfx^{i+1} = \bbfx + \bfK^i(\bfb - \bfh(\bfx^i) - \bfH^i(\bbfx - \bfx^i)),
    \end{equation}
    with
    \begin{subequations}
        \begin{align}
            \bfH^i &= \bfh'(\bfx^i),\\
            \bfK^i &= \bfA(\bfH^{i})^T (\bfH^i \bfA (\bfH^i)^T + \bfB)^{-1},
        \end{align}
    \end{subequations}
    where $\bfh'(\bfx^i)$ denotes the Jacobian of $\bfh$ evaluated at $\bfx^i$.
\end{lemma}
\begin{lemma}\label{lem:jacob}
    We have the following first-order expansion
    \begin{equation*}
       \exp_G(\bar\boxi + \bodelta)\bfd_k = \exp_G(\bar\boxi)\bfd_k+ \exp_G(\bar\boxi)\bfH_k\calJ_r(\bar\boxi)\bodelta+\|\bodelta\|^2,
    \end{equation*}
   with $\bfH_k$ as in~\eqref{eq:jac_formula}, and  where the right Lie Jacobian~\cite{chirikjian2009stochastic} is defined by $\exp_G(\boxi+\bodelta)\approx\exp_G(\boxi)\exp_G(\calJ_r(\boxi)\bodelta)$ neglecting terms of order $\|\bodelta\|^2$, see \eqref{eq:right_jac} for an expression.
\end{lemma}
\begin{proof}
   We have the first-order expansion in $\bodelta$ 
   \begin{subequations}
       \begin{align}
            \exp_G(\bar\boxi + \bodelta)\bfd_k &\approx \exp_G(\bar\boxi )\exp_G(\calJ_r(\bar\boxi )\bodelta)\bfd_k,\\
            &\approx\exp_G(\bar\boxi )\big(\bfd_k + \Lmap{\calJ_r(\bar\boxi )\bodelta} \bfd_k\big),\\
            &\approx\exp_G(\bar\boxi )\big(\bfd_k + \bfH_k   \calJ_r(\bar\boxi )\bodelta\big),
        \end{align} 
   \end{subequations}
    where we applied~\eqref{eq:jac_formula} with $\boxi= \calJ_r(\bar\boxi )\bodelta.$
\end{proof}

Letting $\bfb:=\bfz_k +\bfd_k $, $\bfh(\boxi):=\exp_G(\boxi)\bfd_k$, and $\bar\bfx= \bfzero$, we see objectives~\eqref{eq:IEKF_max_posterior} and~\eqref{optim:bell} coincide, and a direct application of Lemma~\ref{lem:Kalman_upd} and Lemma~\ref{lem:jacob} yields the following proposition.
 
\begin{proposition} \label{prop:IIEKF_GN_seq}
    The sequence of GN updates for the optimization problem~\eqref{eq:IEKF_max_posterior} writes
    \begin{equation}
        \boxi_{k \mid k-1}^{i+1} = \bfK_k^i\big(\bfz_k - \exp_G(\boxi_{k \mid k-1}^i)\bfd_k + \bfd_k + \bfH_k^{i} \boxi_{k \mid k-1}^i\big),\label{seq:eq}
    \end{equation}
   letting $\bfH_k$ be the \IEKF{} standard Jacobian from~\eqref{eq:H_k} and
        \begin{align} 
         \bfH_k^{i} & = \exp_G(\boxi_{k \mid k-1}^i)\bfH_k \calJ_r(\boxi_{k \mid k-1}^i),\\
            \bfK_k^{i} &= \bfP_{k\mid k-1} (\bfH_k^{i})^T \left(\bfH_k^{i}\bfP_{k \mid k-1}(\bfH_k^{i})^T + \hbfN_k\right)^{-1}.
        \end{align}
\end{proposition}

Assuming the GN method converges after $i^\star$ iterations, the state estimate is updated according to
\begin{equation}
    \hbochi_{k\mid k} = \hbochi_{k \mid k-1}\exp_G\left(\boxi_{k\mid k-1}^{i^\star}\right).
\end{equation}
Since Jacobian $\bfH_k$ is independent from the current estimate, the Riccati update does not require any iteration, and the error covariance is updated once and for all as follows
\begin{equation} \label{eq:Riccati_iekf}
    \bfP_{k\mid k} = (\bfI - \bfK_k\bfH_k)\bfP_{k \mid k-1},
\end{equation}
with $ \bfH_k=\bfH_k^0 $ and $ \bfK_k=\bfK_k^0$. The full algorithm is summarized in Algorithm~\ref{alg:IIEKF}. Note that we exactly recover the \IEKF{} of~\cite{barrau2016invariant} if we perform one iteration only.

\subsection{Special case: noise-free measurements} \label{app:nf_IIEKF}
In the limit case of noise-free measurements, \ie, $\bfN_k=\bfzero$, the innovation covariance becomes $\bfS_k=\bfH_k\bfP_{k\mid k-1}\bfH_k^T$ and can cease to be invertible, in which case the Kalman gain $\bfK_k = \bfP_{k\mid k-1}\bfH_k^T\bfS_k^{-1}$ is undefined. This occurs when the perfectly observed directions overlap between subsequent updates for example. This issue is solved   using the "noise-free limit gain" developed in our preliminary conference paper~\cite{goffin2023invariant}:
\begin{subequations}\label{eq:nf_gain}
    \begin{align}
        \bfK_k^\mathrm{nf} &= \lim\limits_{\delta \rightarrow 0} \bfP_{k\mid k-1}\bfH_k^T(\bfH_k\bfP_{k \mid k-1} \bfH_k^T + \delta \bfI)^{-1},\\
        &= \bfL_{k\mid k-1} (\bfH_k\bfL_{k\mid k-1})^\dagger,
    \end{align}
\end{subequations} 
where $(\cdot)^\dagger$ is the Moore-Penrose pseudo-inverse, $\bfP_{k\mid k-1} = \bfL_{k\mid k-1}\bfL_{k\mid k-1}^T$, and $\mathrm{nf}$ stands for "noise-free". This gain is always defined.

Let us analyze the behavior of the proposed \IIEKF{} as the measurement noise magnitude approaches zero. Let $\hbfN_k=\delta\bfI$, with $\delta\ll 1$. In this context, Problem~\eqref{eq:IEKF_max_posterior} becomes that of minimizing $f_\delta(\boxi) = \frac{1}{2}\|\boxi\|^2_{\bfP_{k\mid k-1}} + \frac{1}{2}\|\bfz_{k} - \exp_G(\boxi)\bfd_k + \bfd_k\|^2_{\delta \bfI}$.  
The minimizer of $f_\delta$ may be sought using the GN sequence of estimates of Proposition~\ref{prop:IIEKF_GN_seq}
\begin{align}
    \begin{split}
        &\boxi_{k\mid k-1,\delta}^{i+1} =  \bfK_{k,\delta}^{i}\cdot\\
        & \quad \underbrace{\left(\bfz_k -\exp_G(\boxi_{k\mid k-1,\delta}^{i})\bfd_k + \bfd_k + \bfH_{k,\delta}^{i} \boxi_{k\mid k-1,\delta}^{i}\right)}_{\bfz_{k,\delta}^i},
    \end{split}
\end{align}
where $\bfH_{k,\delta}^{i} = \exp_G(\boxi_{k\mid k-1,\delta}^{i})\bfH_k\calJ_r(\boxi_{k\mid k-1,\delta}^{i})$ and where we let $ \bfK_{k,\delta}^{i}=\bfP_{k\mid k-1}(\bfH_{k,\delta}^{i})^T\left(\bfH_{k,\delta}^{i}\bfP_{k\mid k-1} (\bfH_{k,\delta}^{i})^T + \delta\bfI\right)^{-1}$. 
Starting from $\boxi_{k\mid k-1,\delta}^0 = \bfzero$ and letting $\delta \rightarrow 0$, we get
\begin{align}
        \lim\limits_{\delta \rightarrow 0} \boxi_{k\mid k-1,\delta}^{i+1} &= \lim\limits_{\delta \rightarrow 0} \bfK_{k,\delta}^{i} \cdot \lim\limits_{\delta \rightarrow 0} \bfz_{k,\delta}^{i},
\end{align}
so that the GN sequence approximating the solution to the noise-free optimization problem becomes, recalling \eqref{eq:nf_gain},
\begin{equation}\label{eq:nf_GN_seq}
    \boxi_{k\mid k-1}^{i+1} = \bfL_{k\mid k-1}(\bfH_k^i\bfL_{k\mid k-1})^\dagger \bfz_k^i.
\end{equation}
This provides the  limit \IIEKF{} for noise-free measurements  described in Algorithm~\ref{alg:nf_IIEKF}.

\section{Adaptation and extension to the \LGIEKF{}}\label{Appendix:Bourmaud}
When providing the \LGIEKF{} with features from invariant filtering,  the two key differences between the \LGIEKF{} and the \IIEKF{} lie in their Riccati update: 1) the \LGIEKF{} computes $\bfP_{k\mid k}$ using $\bfK_k^i$ and $\bfH_k^i$, the gain and output Jacobian obtained in the final iteration of the GN method, and 2) it involves $\calJ_r(\boxi_{k \mid k-1}^i)$, the right Jacobian of $G$, to  adjust the linearization of the exponential map when expanding around $\boxi_{k\mid k-1}^i$ instead of $\bfzero$. This Jacobian writes:
\begin{equation}\label{eq:right_jac}
    \Lmap{\calJ_r(\bodelta) \bozeta} := \sum\limits_{k=0}^{+\infty}\frac{[\dots[\Lmap{\bozeta}, \overbrace{\Lmap{\bodelta}],\dots, \Lmap{\bodelta}}^k]}{(k+1)!}.
\end{equation}
Despite these differences, the following result holds.

\begin{corollary} \label{cor:LG-IEKF}
    If noise-free measurements $\bfy_k=\bochi_k \bfd_k$ are considered, uncertainty model \eqref{eq:concentrated_gaussian} is used (i.e., a left-invariant error), and potential rank deficiency issues in the Kalman gain are addressed via the noise-free gain from Algorithm~\ref{alg:nf_IIEKF}, then Theorems \ref{thm:ric} and \ref{thm:update_compatibility} hold, replacing \IIEKF{} with \LGIEKF{}.
    \end{corollary}
\begin{proof}
    As $\bfH_k \calJ_r(\boxi_{k\mid k-1}^i)=\exp_G(-\boxi_{k \mid k-1}^i)\hbochi_{k\mid k-1}^{-1}\bfH_k^i$, see line~11 in Algorithm~\ref{alg:LG-IEKF}, we have the following equivalence: $$\bfH_k\bfP_{k\mid k}\bfH_k^T= \bfzero 	\Leftrightarrow \bfH_k^i(\bfI - \bfK_k^i\bfH_k^i)\bfP_{k\mid k-1}(\bfH_k^i)^T = \bfzero.$$ The right-hand side of this equivalence follows directly from the same reasoning used in the proof of Theorem~\ref{thm:ric}. Regarding implication~\eqref{eq:geometry_respect}, the same reasoning as in the proof of Theorem~\ref{thm:update_compatibility} is used to prove $\hbochi_{k\mid k} \in \obset{\bochi \bbfd_k = \bbfy_k}$. Let us show now that the equality $\bbfH_k\bfP_{k\mid k}\bbfH_k^T=\bfzero$ also holds. The assumption $\bbfH_k\bfP_{k\mid k-1}\bbfH_k^T=\bfzero$ implies that $\Lmap{\boxi} \in \liealg{s}$, for all $\boxi \in \im \bfP_{k\mid k-1}$, where $\liealg{s}$ is the subalgebra of $\liealg{g}$ defined in Proposition~\ref{prop:structure}. As the Lie bracket is closed in $\liealg{s}$, the expression in~\eqref{eq:right_jac} is such that $\bbfH_k\calJ_r(\bodelta)\bozeta = \bfzero$, for all $\bodelta,\bozeta \in \im \bfP_{k\mid k-1}$. 
    Recalling that $\bfP_{k\mid k} = \calJ_r(\boxi_{k \mid k-1}^i) (\bfI - \bfK_k^i\bfH_k^i)\bfP_{k\mid k-1} \calJ_r(\boxi_{k \mid k-1}^i)^T$, 
    that $\boxi_{k \mid k-1}^i \in \im \bfP_{k\mid k-1}$, and 
    that $\im (\bfI - \bfK_k^i\bfH_k^i)\bfP_{k\mid k-1} \calJ_r(\boxi_{k \mid k-1}^i)^T \subseteq \im \bfP_{k\mid k-1}$, 
    we necessarily have $\bbfH_k\bfP_{k\mid k} = \bfzero$, and thus $\bbfH_k\bfP_{k\mid k}\bbfH_k^T=\bfzero$.
\end{proof}
We stress that this result--a consequence of the Lie subalgebra structure related to the observation \eqref{eq:nf_meas}--does not, at any rate, extend to the general \LGIEKF{} presented in~\cite{bourmaud2016intrinsic}.

\bibliographystyle{IEEEtran}
\bibliography{biblio.bib}

\end{document}